\RequirePackage{fix-cm}
\documentclass{svjour3}                     
\smartqed  
\usepackage{graphicx}
\begin{document}

\title{Modeling the dynamics of Hepatitis C Virus with combined antiviral drug therapy: Interferon and Ribavirin. 
}


\author{Sandip Banerjee \and Ram Keval \and S. Gakkhar}


\institute{Sandip Banerjee \at
              Department of Mathematics \\
              Indian Institute of Technology Roorkee (IITR)\\
              Roorkee 247667, Uttaranchal, India.\\
              Tel.: +91-1332-285697\\
              \email{sandofma@iirt.ernet.in}           
           \and
           Ram Keval \at
              \email{ramkeval@gmail.com}
              \and
              S. Gakkhar \at
              \email{sungkfma@iitr.ernet.in}
}

\date{Received: date / Accepted: date}

\maketitle

\begin{abstract}
A mathematical modeling of Hepatitis C Virus (HCV) dynamics has been presented in this paper. The proposed model, which involves four coupled ordinary differential equations, describes the interaction of target cells (hepatocytes), infected cells, infectious virions and non-infectious virions. The model takes into consideration the addition of ribavirin to interferon therapy and explains the dynamics regarding biphasic and triphasic decline of viral load in the model. A critical drug efficiency parameter has been defined and it is shown that for efficiencies above this critical value, HCV is eradicated whereas for efficiencies lower this critical value, a new steady state for infectious virions is reached, which is lower than the previous steady state.

\end{abstract}

\section{Introduction}
\label{intro}
Hepatitis C is an infectious viral disease caused by Hepatitis C Virus (HCV) and its infection route is via the blood. HCV was first identified in 1989, when the expression of cDNAs obtained from the blood plasma of a chimpanzee was induced with hepatitis non-A, non B and was screened with convalescent serum \cite{Choo89}. It should be noted that HCV was the first virus to be discovered by molecular biological method and not by previously used virological methods. Accordingly to WHO, individuals between 200 to 300 millions worldwide are currently infected with HCV \cite{Who00}. With the discovery of HCV virus, a new dimension on the treatment and prevention of liver diseases evolved \cite{Hayashi06}. Previously, patients who were thought to be suffering from hepatitis non-A or non-B or even alcoholic liver disease, actually have hepatitis C. As the dynamics of hepatitis C virus was understood  with time, it becomes clear that the disease is a major threat for hepatocellular carcinoma.\\

Since the identification of the hepatitis C virus, great efforts have been made to counter-attack it with an antiviral therapy. As a crucial mediator of the innate antiviral immune response, interferon-$\alpha$ (IFN-$\alpha$) was a natural choice for treatment \cite{Jordan05}. With IFN-$\alpha$ therapy, there was a rapid decline in HCV-RNA levels in serum but it can achieve only moderate success. With IFN-$\alpha$ monotherapy, that is, therapy will IFN-$\alpha$ alone do not achieve clearance of HCV in half of the infected individuals suffering from chronic hepatitis C. A major break through came with the addition of broad-spectrum antiviral agent ribavirin to IFN-$\alpha$ treatment. With combination of pegyleted interferon and ribavirin, a sustained response rates of 54-56\% can be achieved \cite{Mann01,Fried02,Had04}. But this combined therapy is expensive and have some side effects. Therefore, there is a need for more effective and better tolerated therapies for hepatitis C. However, due to difficulty of developing new, potent, specific agents against HCV, this combined therapy of pegylated interferon and ribavirin is going to remain as the sole therapy in the near future. With more better understanding of the mechanism of interferon and ribavirin together with improvement in dosing and dose regimen, it is expected that there is going to be substantial improvement in the response rates of the patients suffering from hepatitis C.\\

Some patients suffering from hepatitis C naturally clear the virus infection without any medical intervention but mostly (55-85\%) cannot do that and develop chronic HCV infection \cite{Hoof02}. In the standard protocol for the treatment of hepatitis C, a patient is given weekly injection of IFN and take ribavirin pills daily for the period of treatment \cite{Dienst06}. If after six months of therapy, a patient does not show any trace of hepatitis C viral load, then the patient is said to have achieved sustained virological response (SVR), implying clinically cured. The aim of this combined therapy is to reduce the viral load to a minimum, so that a patient achieves SVR. If the therapy fails, a patient fails to achieve SVR and suffers from chronic hepatitis C. It should be noted that in hepatitis C patients, the rate of achievements of SVR depends on the genotype of the particular infecting virus. Genotypes 1 and 4 are usually treated for 48 weeks whereas genotypes 2 and 3 are done for 24 weeks with lesser doses of ribavirin \cite{Stra04}. It is observed that 46\% of the patients with genotype 1 have achieved SVR under this treatment whereas SVR rates of 76-82\% are observed in patients with HCV genotypes 2 or 3 \cite{Sulk03}.\\

To understand the dynamics of hepatitis C, specially, to determine the efficiency of IFN-$\alpha$ as monotherapy and also in combination with ribavirin, mathematical models have been used extensively, which have provided insights into the pathogenesis of HCV in vivo \cite{Neum98,Dixit04,Dahari07a,Dahari07b,Chakra08}. Other than suggesting predominated mechanism(s) of drug action, the models have been successful in explaining the confounding patterns of viral load changes in HCV infected patients undergoing therapy.\\

Neumann et. al \cite{Neum98}, used the basic model of viral dynamics to HCV, which assumes a simplified view of HCV infection by taking into account of IFN-$\alpha$ monotherapy. The model consists of three populations, namely, uninfected hepatocytes (the target cell), productively infected hepatocytes and free HCV virions and describes the response to interferon therapy.  Though the basic model by Neumann \cite{Neum98} provides insights into the HCV dynamics in vivo but it fails to predict the long term response rates observed during combination therapy as it ignores the influence of ribavirin. Dixit et. al. \cite{Dixit04} made an advancement in the basic model by taking into account the role of ribavirin action along with interferon. The dynamics of model suggests that ribavirin has very little role to play during the first phase decline induced by interferon, as ribavirin does not alter the viral production. However, ribavirin enhances the second phase slope when interferon effectiveness is small, which is in close agreement with experiments \cite{Stefan02} .\\

One of the drawback of the model by Dixit et. al. \cite{Dixit04} is that non-responders and patients with triphasic decay patterns are not studied and is restricted to biphasic responses only. A modified model, which include the proliferation of uninfected and infected cells driven by liver homeostatic mechanisms was studied by Dahari et. al. \cite{Dahari07a,Dahari07b}. The model predicts the triphasic decline and succeeds in explaining the origins of non-response. In a recent study, Chakrabarty and Joshi \cite{Chakra08} presented a model where they used deterministic control theory to obtain an optimal treatment strategy using interferon and ribavirin. The optimal treatment which they obtained succeeded in reducing the levels of viral load and at the same time keeping the side effects of the drug to a minimal. Swati DebRoy et al. \cite{Swati10} discussed one of the common and life threatening side effect, namely, hemolytic anemia, of hepatitis C infection. They used an extension and modification of Neumann's model \cite{Neum98} to study the effect of combination therapy in the light of anemia and succeeded in providing a quantification of the amount of drug a body can tolerate without succumbing to hemolytic anemia.\\

In this paper, a mathematical model is proposed depicting the behavior of hepatitis C virus. The idea is to capture the dynamics of the model due to combined effect of antiviral drug therapy: interferon and ribavirin. In section 2, formulation of the model based of the schematic diagram has been discussed. Qualitative analysis of the model has been shown in section 3, which includes determination of equilibrium points, positivity and boundedness of the system, local stability analysis and persistence of the system. Global stability analysis about the endemic equilibrium point has been shown in section 4. Numerical results, capturing the dynamics of HCV are discussed in section 5. The paper ends with a discussion.

\section{Mathematical Model}
Adaptive immune responses mediated by T cells are essential in the control of HCV and viral clearance. Recent studies in which memory CD4+ and CD8+ T cells were depleted have confirmed the critical role of these cells in controlling HCV infections \cite{Grakoui03}. Hence, these cells are the target cells (or hepatocytes). Once infected with HCV, these target cells becomes infected hepatocytes, which replicates the hepatitis C virus. The replicated viruses may be infectious or non-infectious, depending on the effectiveness of the drug(s) therapy. Figure 1 gives the schematic representation of the above biological scenario. The proposed model is given by the following system of coupled ordinary differential equations:
\begin{eqnarray}
\label{eqn-1-hep}
\frac{dT}{dt}&=&s + r T\left(1-\frac{T+I}{k}\right) -d_{1} T-(1-c \eta_{1})\alpha V_{I} T\\
 \frac{dI}{dt}&=&(1-c \eta_{1})\alpha V_{I} T - d_{2}I \\
 \frac{dV_{I}}{dt}&=&\left(1-\frac{\eta_{r}+\eta_{1}}{2}\right)\beta I-d_{3}V_{I}\\
  \frac{dV_{NI}}{dt}&=&\left(\frac{\eta_{r}+\eta_{1}}{2}\right)\beta I-d_{3}V_{NI}
\end{eqnarray}
The first equation gives the dynamics of the target cells, that is, hepatocytes. They are produced from a source at a constant rate $s$ and at the same time, their growth is augmented by a logistic term with an intrinsic growth rate $r$ and carrying capacity $k$. The hepatocytes die naturally at a rate $d_1$. The logistic term $rT\left(1-\frac{T+I}{k}\right)$ is more realistic in the sense that it incorporates total hepatocytes, both non-infected and infected ones. Here $\eta_{1}$ denotes the effectiveness (or efficacy) of interferon in blocking the release of new virions, $c \eta_{1}$ is fraction of the efficacy $(0 < c < 1)$ and hence $(1-c \eta_{1})$ gives the ineffectiveness of interferon, which fails to stop the target cells getting infected. The number of cells which gets infected is proportional to the number of infection virions and available target cells (hepatocytes) with a proportionality constant $\alpha$, which explains the term $(1-c \eta_{1})\alpha V_{I} T$ in (1) and (2). Infected cells die at a rate $d_{2}$. Equations (3) and (4) give the dynamics of virions, namely, the infected and non infected ones. The term $\left(\frac{\eta_{r}+\eta_{1}}{2}\right)$ is the effectiveness of combined effect of interferon and ribavirin $\left(0 < \eta_{r} < 1, 0 < \eta_{1} <1 \Rightarrow 0 < \frac{\eta_{r}+\eta_{1}}{2} <1\right)$ and hence $\left(1-\frac{\eta_{r}+\eta_{1}}{2}\right)$ is the ineffectiveness of the combined effect, which leads to the growth of the infection virions, proportional to the number of infected cells $(I)$ with proportionality constant $ \beta$ (infection rate). This explain the term $\left(1-\frac{\eta_{r}+\eta_{1}}{2}\right) \beta I$ in (3). The infectious virions die at a rate $d_{3}$. In (4), the combined effect of interferon and ribavirin results in virions $(V_{NI})$, which are non-infectious in nature and also die at a rate $d_{3}$. System (1-4) has to be analyzed with the following initial conditions: $ T(0) > 0, I(0) > 0, V_{I}(0) > 0, V_{NI}(0) > 0.$ All system parameters are positive.
\section{Qualitative Analysis of the model}
\subsection{\textbf{Positivity and Boundedness}}
\begin{theorem}\label{th1}
The solutions of the system (1-4) are positive for all t $>$ 0.
\end{theorem}
\begin{proof}
Let $(T(t), I(t), V_{I}(t), V_{NI}(t))$ be a solution of system (1-4). Consider for some time $t_{1} > 0, T(t_{1})=0$ and if $t_{1}$ be the first such time, then $\frac{d T(t_{1})}{d t}\leq 0.$ But from (1), $\frac{d T(t_{1})}{d t} = s > 0$, which is a contradiction. Therefore, $T(t)>0,$ for all $t > 0$.\\
It is again observed that $I(t)>0 ~\textrm{for}~ t \geq 0$ as
\begin{eqnarray}
\frac{d I}{d t} \geq - d_{2}I \Rightarrow I(t)=I(0) e^{-d_{2}t}\\
\end{eqnarray}

Similarly, (3) and (4) gives
\begin{eqnarray}
V_{I}&=&\left(1-\frac{\eta_{r}+\eta_{1}}{2}\right)\beta e^{-d_{3}t}\int\limits_{0}^{t}I(s) e^{-d_{3}s}d s >0\\ \and
V_{NI}&=&\left(\frac{\eta_{r}+\eta_{1}}{2}\right)\beta e^{-d_{3}t}\int\limits_{0}^{t}I(s) e^{-d_{3}s}d s >0,
\end{eqnarray}
Therefore, it is concluded that $(T(t), I(t), V_{I}(t), V_{NI}(t)) > 0$ for all $t \geq 0.$
\end{proof}
\begin{theorem}\label{th1}
The solutions of the system (1-4) are bounded.
\end{theorem}
\begin{proof}
Let $(T(t), I(t), V_{I}(t), V_{NI}(t))$ be any solution of system. Let
\begin{equation}
\rho(t) = T(t)+I(t)+V_{I}(t)+V_{NI}(t)
\end{equation}
The derivative of $\rho(t)$ along the positive solutions of system \ref{eqn-1-hep} is given by
\begin{eqnarray}\label{eqn-2-hep}
\nonumber \frac{d\rho}{dt}&=& s + r T-\frac{r T^{2}}{k}-\frac{r T I}{k}-d_{1} T-d_{2}I+\beta I-d_{3}V_{I}-d_{3}V_{NI}\\
 \nonumber &\leq & s + r T-\frac{r T^{2}}{k}-d_{1} T-d_{2}I+\beta I-d_{3}V_{I}-d_{3}V_{NI}\\
\nonumber & \leq & s + r T-\frac{r T^{2}}{k}- A \rho(t)\\
\nonumber &=& - A \rho(t)+(s+\frac{r k}{4}) -\left(\frac{\sqrt{r}T}{\sqrt{k}}-\frac{\sqrt{r k}}{2}\right)^{2}\\
 &\leq &- A \rho(t)+(s+\frac{r k}{4})
\end{eqnarray}
where $A = \min [d_{1}, (d_{2}-\beta), d_{3}]$ and $d_2>\beta$. It follows from (\ref{eqn-2-hep}) that \[\lim_{t \to +\infty}\sup \rho(t)\leq (s+\frac{r k}{4}) = M^{*}.\] Therefore, there exist positive constants $M > M^{*}$ and $T_{1} > 0$ such that if $t\geq T_{1}, \rho(t)< M.$ This completes the proof.
\end{proof}

\subsection{\textbf{Equilibria}}
System (1-4) has the following non-negative equilibrium points:\\
  (a)  $E_{0}(\widehat{T}, 0, 0, 0)$  where $\widehat{T}=\frac{k}{2 r} \left[(r-d_{1})+\sqrt{(r-d_{1})^2+\frac{4 r s}{k}} \right]$, where $r > d_{1}$. We assume $s\leq k d_1$ so that the model is physiologically realistic, that is, $\widehat{T} \leq k$.\\
  (b)  $E_{*}(T^{*},I^{*},V_{I}^{*},V_{NI}^{*})$, where,

\begin{eqnarray*}
T^{*} = \frac{d_{2}d_{3}}{(1-c \eta_{1})\left(1-\frac{\eta_{r}+\eta_{1}}{2}\right)\alpha \beta},\\
I^{*}=\frac{s r \hat{R}^{2} + (r-d_{1})^{2} k (\hat{R}-1)}{r \hat{R}[d_{2} \hat{R}+(r-d_{1})]},\\
V_{I}^{*}=\frac{\left(1-\frac{\eta_{r}+\eta_{1}}{2}\right)\beta I^{*}}{d_{3}},\\
V_{NI}^{*} =\frac{\left(\frac{\eta_{r}+\eta_{1}}{2}\right)\beta I^{*}}{d_{3}}.
\end{eqnarray*}

Here, $\hat{R}$, the controlled reproductive number (CRN) \cite{Swati10} is defined as $\hat{R}=\frac{R_{0}  k}{\widehat{T}}\left(\frac{r-d_{1}}{r}\right)$  and the basic reproductive number of the model is calculated as $R_{0}=\frac{\widehat{T}}{T^{*}}$. It may be noted that $E_{*}$ exists even when $\hat{R}<1$, provided, $s r \hat{R}^{2} + (r-d_{1})^{2} k (\hat{R}-1)>0$. However, it always exists for $\hat{R}>1$

\subsection{\textbf{Local Stability Analysis}}
\begin{theorem}
The disease free equilibrium point $E_0$ is locally asymptotically stable if $R_0 < 1$ and unstable if $R_0 > 1$.
\end{theorem}
\begin{proof}
The variational matrix for the equilibrium point $E_{0}$ is \\
$V_{0}=\left(
  \begin{array}{cccc}
     (r-d_{1})-\frac{2 r \widehat{T}}{k} &-\frac{r \widehat{T} }{k}& -(1-c \eta_{1})\alpha  \widehat{T}  & 0 \\
    0 & -d_{2} & (1-c \eta_{1})\alpha  \widehat{T}& 0 \\
    0 & \left(1-\frac{\eta_{r}+\eta_{1}}{2}\right)\beta & -d_{3} & 0 \\
    0 & \left(\frac{\eta_{r}+\eta_{1}}{2}\right)\beta & 0 & -d_{3}\\
  \end{array}
\right)\\\\\nonumber$\\
The corresponding characteristic equation is
\begin{equation}\label{eqn-3-hep}
\left(\lambda-(r-d_{1})+\frac{2 r \widehat{T}}{k}\right)(\lambda+d_{3})(\lambda^{2}-C \lambda+D)=0
\end{equation}
where $C =-(d_{2}+d_{3})$ and $D=[d_{2}d_{3}-\frac{1}{2}(1-c \eta_{1})(2-\eta_{r}-\eta_{1})\alpha\beta \widehat{T}]$. For local asymptotic stability, all roots of the characteristic equation must be negative or have negative real parts. Since, the linear factor gives negative eigenvalues and $C<0$, for stability $D$ must be positive, that is, $\frac{\widehat{T}}{T^{*}} < 1$. Therefore, the uninfected steady state $E_{0}$ is stable if $R_{0}< 1$ and unstable if $R_{0}>1$.
\end{proof}
\begin{theorem}
The endemic equilibrium point $E_{*}$ is locally asymptotically stable if $1 < \widehat{R} <3$.
\end{theorem}

\begin{proof}

The variational matrix for the equilibrium point $E_{*}$ is

$V_{2}=\left(
  \begin{array}{cccc}
    -\frac{s r \hat{R}}{k(r-d_{1})}-\frac{(r-d_{1})}{\hat{R}}& -\frac{(r-d_{1})}{\hat{R}} & -\frac{(1-c \eta_{1})\alpha k(r-d_{1})}{r \hat{R}}& 0 \\
   d_{2}[\frac{s r \hat{R}^2+k(r-d_{1})^2(\hat{R}-1)}{(r-d_{1})k(d_{2}\hat{R}+(r-d_{1}))}]& -d_{2} & \frac{(1-c \eta_{1})\alpha k(r-d_{1})}{r \hat{R}} & 0 \\
    0 & \left(1-\frac{\eta_{r}+\eta_{1}}{2}\right)\beta & -d_{3} & 0 \\
    0 & \left(\frac{\eta_{r}+\eta_{1}}{2}\right)\beta & 0 & -d_{3}\\
  \end{array}
\right)$\\

The characteristic equation is obtained as
\begin{equation}\label{eqn-4-hep}
(\lambda_{1}+ d_{3}) (\lambda^{3}+A_{1} \lambda^{2}+ A_{2}\lambda + A_{3})=0
\end{equation}
where
\begin{eqnarray*}
A_{1} &=&d_{2}+d_{3}+\frac{(r-d_{1})}{\hat{R} }+\frac{r \hat{R} s}{(r-d_{1}) k}\\
A_{2} &=& (d_{2}+d_{3})\left[\frac{(r-d_{1})}{\hat{R}}+\frac{r \hat{R} s}{(r-d_{1}) k}\right]\\
&+&\frac{d_{2}}{((r-d_{1})+d_{2}\hat{R})}\left[(r-d_{1})^{2} \left(1-\frac{1}{\hat{R} }\right)+\frac{r \hat{R} s}{k}\right]\\
A_{3} &=&d_{2}d_{3}\left[(r-d_{1})\left(1-\frac{1}{\hat{R}}\right)+\frac{r \hat{R}s}{ (r-d_{1}) k}\right]\\
A_{1} A_{2}-A_{3}&=&\left(\frac{3}{\hat{R}}-1\right)d_{2} d_{3} r+(d_{2}^{2}+d_{3}^{2})\frac{(r-d_{1})}{\hat{R}}\\
&+&(d_{2}+d_{3})\left[\frac{(r-d_{1})^{2}}{\hat{R}^{2}}+\frac{2 r s}{k}+\frac{r^{2}\hat{R}^{2}s^{2}}{(r-d_{1})^2 (k)^{2}}\right]\\
&+& \frac{d_{2} (r-d_{1})}{(r-d_{1})+d_{2}\hat{R}}\left[ (r-d_{1})d_{2}+(r-d_{1})d_{3}+\frac{(r-d_{1})^{2}}{\hat{R}}\right]\left(1-\frac{1}{\hat{R}}\right)\\
 &+&d_{2}(d_{2}+d_{3})\frac{r \hat{R} s}{k}\left[\frac{1}{(r-d_{1})}+\frac{1}{(r-d_{1})+d_{2}\hat{R}}\right]+\frac{r \hat{R} s d_{3}^{2}}{(r-d_{1}) k}\\
&+&\frac{d_{2}r \hat{R} s}{k((r-d_{1})+d_{2}\hat{R})}\left[ (r-d_{1})+\frac{r \hat{R} s}{k(r-d_{1})}\right]
\end{eqnarray*}

According to Routh-Hurwitz criterion, the necessary and sufficient conditions for all the roots of cubic equation (\ref{eqn-4-hep}) to have negative real parts are $A_{1}> 0,A_{2}> 0,A_{3}> 0$ and $A_{1}A_{2}-A_{3}> 0.$ Clearly, $A_{1}>0$. Now,  $A_{2}> 0$ and $A_{3}> 0$ if $\hat{R}>1$ and $A_{1}A_{2}-A_{3}> 0$ provided  $1 < \hat{R}< 3.$  Thus, the endemic equilibrium point $E_{*}(T^{*},I^{*},V_{I}^{*},V_{NI}^{*})$ is locally asymptotically stable if the controlled reproductive number $\hat{R}$ satisfies the condition. It should be noted that the endemic equilibrium point $E_{*}$ will be unstable whenever it exists for $\hat{R}< 1$.
\end{proof}
\subsection{\textbf{Critical drug efficacy}}
In system (1), the effectiveness of interferon and ribavirin are given by the terms $(1-c \eta_1)$ and $\left(1-\frac{\eta_{r}+\eta_{1}}{2}\right)$. Here, these terms are combined into a single term as $(1-c \eta_1) \left(1-\frac{\eta_{r}+\eta_{1}}{2}\right) = 1 - \eta$, where $\eta$ represents the overall drug efficiency. Clearly, $\eta_1 = \eta_r = 0$, that is, drug efficiency is zero before treatment and $0 < \eta_1 < 1, 0 < \eta_r < 1$ during antiviral therapy. The stability criteria for uninfected steady state ($D>0$) is
\begin{eqnarray*}
d_{2}d_{3}-\frac{1}{2}(1-c \eta_{1})(2-\eta_{r}-\eta_{1})\alpha \beta \widehat{T} > 0, \textrm{where},~ \widehat{T}=\frac{k}{2 r} \left[(r-d_{1})+\sqrt{(r-d_{1})^2+\frac{4 r s}{k}} \right]\\
\Rightarrow~ (1-c \eta_1) \left(1-\frac{\eta_{r}+\eta_{1}}{2}\right) < \frac{d_{2}d_{3}}{\alpha \beta \widehat{T}}\\
\Rightarrow~ 1 - \eta < \frac{2 r d_{2}d_{3}}{\alpha \beta [(r-d_1) k + \sqrt{(r-d_1)^2 k^2+ 4 r s k)}]}
\end{eqnarray*}

This clearly indicates that there exists a point that acts as a point of separation between the region of stability for uninfected steady state and the region of stability for infected steady state. This point can be termed as a transcritical bifurcation point, which is given by

\begin{eqnarray*}
1 - \eta = \frac{2 r d_{2}d_{3}}{\alpha \beta [(r-d_1) k + \sqrt{(r-d_1)^2 k^2+ 4 r s k)}]}
\end{eqnarray*}

Accordingly, the critical drug efficacy is defined as
\begin{eqnarray*}
\eta_{\varepsilon} = 1 - \frac{2 r d_{2}d_{3}}{\alpha \beta [(r-d_1) k + \sqrt{(r-d_1)^2 k^2+ 4 r s k)}]} = 1 - \frac{T_0^{*}}{\widehat{T}}
\end{eqnarray*}

where $T_{0}^{*} = \frac{d_2 d_3}{\alpha \beta}$ is the number of uninfected hepatocytes in an infected person before treatment (obtained by putting $\eta_1 = \eta_r = 0$ in $T^*$) and $\widehat{T}$ is the total number of hepatocytes in an uninfected individual.

During antiviral therapy, if $\eta > \eta_{\varepsilon}$, the drug therapy is successful and the viral load can be eradicated. On the other hand, if $\eta < \eta_{\varepsilon}$, the viral load and the infected cells converge to a new steady state with lower values.

\subsection{\textbf{Permanence of the system}}
\begin{definition}
 A system is said to be permanent if there exist positive constants $\delta, \Delta,$ with $0<\delta\leq\Delta$ such that\\
 \[\min \left\{\lim_{x \to +\infty}\inf T(t),\lim_{x \to +\infty}\inf I(t),\lim_{x \to +\infty}\inf V_{I}(t),\lim_{x \to +\infty}\inf V_{NI}(t)\right\}\geq \delta, \]
\[ \max \left\{\lim_{x \to +\infty}\sup T(t),\lim_{x \to +\infty}\sup I(t),\lim_{x \to +\infty}\sup V_{I}(t),\lim_{x \to +\infty}\sup V_{NI}(t)\right\}\leq \Delta \] for all solutions of the system.
\end{definition}

\begin{theorem}\label{th1}
 The system (1-4) is permanent if $(r-d_{1}) -(1-c \eta_{1})\alpha \overline{V_{I}}-\frac{r \overline{I}}{k}>0$.
\end{theorem}\label{th1}

\begin{proof}
Since all the state variables are bounded, it is obvious that $I(t)\leq M^{**}.$ Then (1) can be rewritten as
\begin{eqnarray*}
\frac{dT}{dt}&=& s + r T\left(1-\frac{T+I}{k}\right)-d_{1} T-(1-c \eta_{1})\alpha V_{I} T \leq s + r T - \frac{r T^{2}}{k}\\
&=& s - \left(\frac{\sqrt{r}T}{\sqrt{k}}-\frac{\sqrt{r k}}{2}\right)^{2}+\frac{r k}{4}\\
& \leq &  \left(s+ \frac{r k}{4}\right).\\
\Rightarrow~~\lim_{t \to +\infty}\sup T(t)&\leq& \left(s+\frac{r k}{4}\right) = \overline{T}~~ (say).
\end{eqnarray*}

(3) and (4) gives, \\
\[\frac{d(V_{I}+V_{NI})}{dt}= \beta I-d_{3}(V_{I}+V_{NI})\leq d_{3}\left(\frac{\beta M^{**}}{d_{3}}-(V_{I}+V_{NI})\right).\]\\
Hence,  \[\lim_{t \to +\infty}\sup (V_{I}+V_{NI})\leq \frac{\beta M^{**}}{d_{3}}.\]
  Therefore, \[\lim_{t \to +\infty}\sup V_{I}\leq \frac{\beta M^{**}}{d_{3}}=\overline{V_{I}}~~(say)\] and  \[\lim_{t \to +\infty}\sup V_{NI}\leq \frac{\beta M^{**}}{d_{3}}=\overline{V_{NI}} ~~(say).\]\\
From (2) we get, \\
\[\frac{dI}{dt}=(1-c \eta_{1})\alpha V_{I} T - d_{2}I \leq (1-c \eta_{1}) \overline{V_{I}} \alpha \overline{T}- d_{2}I =d_{2}\left[\frac{(1-c \eta_{1}){{\overline{V_{I}} \alpha \overline{T}}}}{d_{2}}-I \right] .\]
So \[\lim_{t \to +\infty}\sup I(t)\leq \frac{(1-c \eta_{1})\overline{V_{I}} \alpha \overline{T}}{d_{2}}=\overline{I}.\]
From the definition of limit superior, again for every $\epsilon > 0,$ there exist a time $ T_{2}$ such that $I(t)\geq (\overline{I}-\epsilon)$ for $t\geq T_2.$

Again from (3) we get,
\begin{eqnarray*}
\frac{dV_{I}}{dt}&=& \left(\frac{2-\eta_{1}-\eta_{r}}{2}\right)\beta I-d_{3}V_{I}\geq \left(\frac{2-\eta_{1}-\eta_{r}}{2}\right)\beta(\overline{I}-\epsilon)-d_{3}V_{I}\\
&=& d_{3}\left[\frac{(2-\eta_{1}-\eta_{r}) \beta (\overline{I}-\epsilon)}{2 d_{3}}-V_{I}\right].
 \end{eqnarray*}
Using differential equality we have, \[\lim_{t \to +\infty}\inf V_{I}(t) \geq \frac{(2-\eta_{1}-\eta_{r}) \beta (\overline{I}-\epsilon)}{2 d_{3}},\] Since $  \epsilon > 0$ is arbitrarily small, we have\[\lim_{t \to +\infty}\inf V_{I}(t) \geq \frac{(2-\eta_{1}-\eta_{r}) \beta \overline{I}}{2 d_{3}} = \underline{V_{I}}.\]
From (4) we get,
\begin{eqnarray*}
\frac{dV_{NI}}{dt}&=& \left(\frac{\eta_{1}+\eta_{r}}{2}\right)\beta I-d_{3}V_{NI}\geq\left(\frac{\eta_{1}-\eta_{r}}{2}\right)\beta(\overline{I}-\epsilon)-d_{3}V_{NI}\\
&=& d_{3}\left[\frac{(\eta_{1}+\eta_{r}) \beta (\overline{I}-\epsilon)}{2 d_{3}}-V_{NI}\right].
\end{eqnarray*}
Using differential equality we have, \[\lim_{t \to +\infty}\inf V_{NI}(t) \geq \frac{(\eta_{1}+\eta_{r}) \beta (\overline{I}-\epsilon)}{2 d_{3}},\]
Since $  \epsilon > 0$ is arbitrarily small, we have \[\lim_{t \to +\infty}\inf V_{NI}(t) \geq \frac{(\eta_{1}+\eta_{r}) \beta \overline{I}}{2 d_{3}} = \underline{V_{NI}}.\]
Again from (1),
\begin{eqnarray*}
\frac{dT}{dt}&=& s + r T\left(1-\frac{T+I}{k}\right)-d_{1} T-(1-c \eta_{1})\alpha V_{I} T\\
 &\geq & T \left[(r-d_{1}) -\frac{r I}{k}-(1-c \eta_{1})\alpha V_{I} -\frac{r T}{k}\right]\\
 & \geq & T \left[(r-d_{1}) -\frac{r \overline{I}}{k}-(1-c \eta_{1})\alpha \overline{V_{I}} -\frac{r T}{k}\right]\\
&=&\frac{r}{k} T \left[((r-d_{1}) -(1-c \eta_{1})\alpha \overline{V_{I}}-\frac{r \overline{I}}{k})\frac{k}{r} -T \right].
\end{eqnarray*}
Therefore, \[\lim_{t \to +\infty}\inf T(t)\geq \left[(r-d_{1}) -(1-c \eta_{1})\alpha \overline{V_{I}}-\frac{r \overline{I}}{k}\right]=\underline{T}.\]
if
\begin{equation}\label{eqn-5-hep}
\left[(r-d_{1}) -(1-c \eta_{1})\alpha \overline{V_{I}}-\frac{r \overline{I}}{k}\right]>0.
\end{equation}

From (2) we get,
\begin{eqnarray*}
\frac{dI}{dt}&=&(1-c \eta_{1})\alpha V_{I} T - d_{2}I\geq (1-c \eta_{1})\alpha \underline{V_{I}}\underline{T}- d_{2}I =d_{2}[\frac{(1-c \eta_{1})\alpha \underline{V_{I}}\underline{T}}{d_{2}}-I ].
\end{eqnarray*}
which implies \[\lim_{t \to +\infty}\inf I(t)\geq\frac{(1-c \eta_{1})\alpha\underline{V_{I}}\underline{T}}{d_{2}}=\underline{I}.\]
Therefore, from condition (\ref{eqn-5-hep}) we conclude that  $\underline{T}>0 $, which ensures that $ \underline{I}> 0, \underline{V_{I}}> 0, \underline{V_{NI}} > 0.$ This completes the proof of Theorem.
\end{proof}

\section{\textbf{Global Stability of the endemic equilibrium point $E_{*}$}}

To study the global stability of endemic equilibrium point $E_*$, the geometric approach of Li and Muldowney \cite{Li96}, which guarantees the global stability of endemic equilibrium point by obtaining simple sufficient conditions, has been used.

\begin{theorem}
Let $ x\rightarrow f(x) \in R^{4}$ be a $C^{1}$ function (class of functions whose derivatives are continuous) for $x$ in a simply connected domain $D \subset R^{4},$ where\\\\
$x=\left(
  \begin{array}{c}
    T \\
    I \\
    V_{I} \\
    V_{NI}\\
  \end{array}
\right)$
and
$f(x)=\left(
  \begin{array}{c}
    s + r T\left(1-\frac{T+I}{k}\right)-d_{1} T-(1-c \eta_{1})\alpha V_{I} T \\
    (1-c \eta_{1})\alpha V_{I} T - d_{2}I \\
    \left(1-\frac{\eta_{r}+\eta_{1}}{2}\right)\beta I-d_{3}V_{I} \\
    \left(\frac{\eta_{r}+\eta_{1}}{2}\right)\beta I-d_{3}V_{NI} \\
  \end{array}
\right)
$\\\\\\
Consider the system of differential equations $\dot{x}=f(x)$ subject to the initial condition $(T_{0}, I_{0}, V_{I0}, V_{NI0})^{T}= x_0 (say)$
Let $x(t, x_{0})$ be a solution of the system. The system (1-4) has a unique endemic equilibrium point $E_{*}$ in $D$ and there exits a compact absorbing set $K \subset D.$ It is further assumed that the system (1-4) satisfies Bendixson criterion \cite{Li96}, that is robust under $C^{1}$ local perturbations of $f$ at all non-equilibrium non-wandering points of the system. Let $ x \rightarrow M(x)$ be a $6\times6$ matrix valued function that is $C^{1}$ for $x \in D$. It is also assumed that $M^{-1}(x)$ exists and is continuous for $x \in K$. Then the endemic unique equilibrium point $E_{*}$ is globally stable in $D$ if \begin{equation}\label{eqn-31-hep}
  \bar{q}_{2} = \lim_{t \rightarrow \infty} \sup \sup_{x_{0}\epsilon k}\frac{1}{t}\int_{0}^{t} \mu(B(x(s, x_{0})))ds < 0
\end{equation}
where $B=M_{f}M^{-1}+M \frac{\partial f^{[2]} }{\partial x} M^{-1},$ the value $M_{f}$ is obtained by replacing each entry $m_{ij}$ in $M$ by its directional derivative in the direction of $f, \nabla m^{*}_{ij} f$ and $\mu(B)$ is the Lozinski$\breve{i}$ measure of $B$ with respect to a vector norm $|.|$ in $R^{4}$, defined by \cite{Coppel65},

\begin{equation}\label{eqn-31-hep}
  \mu(B) = \lim_{h \rightarrow 0^{+}} \frac{|I+h B|-1}{h}
\end{equation}
\end{theorem}

\begin{proof}
Since permanence with respect to a set of variables implies persistence with respect to the same set of variables, it can be concluded that the system (1-4) is persistent. Clearly, the system (1-4) has a unique endemic equilibrium point $E_*$ in D and the persistence of the system, together with the boundedness of solutions, implies the existence of a compact absorbing set $K \subset D$ \cite{Butler86}.

The Jacobian matrix J of the system (1-4) is\\\\
$J=\left(
  \begin{array}{cccc}
    -\frac{r T}{k}-\frac{s}{ T}& -\frac{r T}{k} & -(1-c \eta_{1})\alpha T& 0 \\
    (1-c \eta_{1})\alpha V_{I}& -d_{2} & (1-c \eta_{1})\alpha T & 0 \\
    0 & (\frac{2-\eta_{r}-\eta_{1}}{2})\beta & -d_{3} & 0 \\
    0 & (\frac{\eta_{r}+\eta_{1}}{2})\beta & 0 & -d_{3}\\
  \end{array}
\right)$\\\\\\
and the corresponding associated second compound matrix $J^{[2]}$ (for detailed discussion of compound matrices, their properties and their relations to differential equations, the readers are referred to \cite{Fiedler74} and \cite{Muldowney90}) is given by

$J^{[2]}=\left(
  \begin{array}{cccccc}
  a_{11}&(1-c \eta_{1})\alpha T & 0 & (1-c \eta_{1})\alpha T & 0 & 0\\
  (\frac{2-\eta_{r}-\eta_{1}}{2})\beta &a_{22}&0& -\frac{r T}{k}&0&0\\
  (\frac{\eta_{r}+\eta_{1}}{2})\beta&0&a_{22}&0&-\frac{r T}{k}&-(1-c \eta_{1})\alpha T\\
  0 &(1-c \eta_{1})\alpha V_{I}&0& -d_{2}-d_{3}&0&0\\
  0&0&(1-c \eta_{1})\alpha V_{I}&0&-d_{2}-d_{3}&(1-c \eta_{1})\alpha T\\
  0&0&0& -(\frac{\eta_{r}+\eta_{1}}{2})\beta&(\frac{2-\eta_{r}-\eta_{1}}{2})\beta&-2d_{3}\\
  \end{array}
\right)$\\\\\\
Set the function $Q = Q(T, I, V_{I},V_{NI}) \equiv   $ diag $\left(1, 1, 1, 1, \frac{I}{V_{I}}, \frac{I}{V_{I}}\right)$, then\\
$Q_{f} Q^{-1 }=  $ diag $\left(0,0,0,0,\frac{\dot{I}}{I}-\frac{\dot{V}_{I}}{V_{I}},\frac{\dot{I}}{I}-\frac{\dot{V}_{I}}{V_{I}}\right)$ and
$B = Q_{f} Q^{-1}+Q J^{[2]}Q^{-1}$\\
Therefore, \\\\
$B = \left(
  \begin{array}{cccccc}
  a_{11}&(1-c \eta_{1})\alpha T & 0 & (1-c \eta_{1})\alpha T & 0 & 0\\
  (\frac{2-\eta_{r}-\eta_{1}}{2})\beta &a_{22}&0& -\frac{r T}{k}&0&0\\
  (\frac{\eta_{r}+\eta_{1}}{2})\beta&0&a_{22}&0&-\frac{r T V_{I}}{I k}&-\frac{(1-c \eta_{1})\alpha T V_{I}}{I} \\
  0 & (1-c \eta_{1})\alpha V_{I}&0&-d_{2}-d_{3} &0&0\\
  0&0&(1-c \eta_{1})\alpha I&0&a_{55}&(1-c \eta_{1})\alpha T\\
  0&0&0& -(\frac{\eta_{r}+\eta_{1}}{2})\frac{\beta I}{V_{I}}&(\frac{2-\eta_{r}-\eta_{1}}{2})\beta&a_{66}\\
  \end{array}
\right)$\\\\\\
$ \equiv \left[
  \begin{array}{cccc}
  B_{11} & B_{12}\\
  B_{21} & B_{22}\\
  \end{array}
\right].$ \\\\\\
where
\begin{eqnarray*}
a_{11} &=& -\frac{r T}{k}-\frac{s}{ T}-d_{2},\\
a_{22} &=& -\frac{r T}{k}-\frac{s}{ T}-d_{3},\\
a_{55} &=& -d_{2}-d_{3}+\frac{\dot{I}}{I}-\frac{\dot{V}_{I}}{V_{I}},\\
a_{66} &=& -2d_{3}+\frac{\dot{I}}{I}-\frac{\dot{V}_{I}}{V_{I}},
\end{eqnarray*}
$B_{11} = [a_{11}]=[-\frac{r T}{k}-\frac{s}{ T}-d_{2}],$
$B_{12}= \left[
  \begin{array}{ccccc}
    (1-c \eta_{1})\alpha T & 0 & (1-c \eta_{1})\alpha T & 0 & 0\\
       \end{array}
\right]$,\\
$B_{21} = [(\frac{2-\eta_{r}-\eta_{1}}{2})\beta, (\frac{\eta_{r}+\eta_{1}}{2})\beta, 0,0,0]^{T}$ and \\
$B_{22} = \left[
  \begin{array}{ccccc}
 a_{22}&0& -\frac{r T}{k} & 0 & 0\\
 0 & a_{22}&0&-\frac{r T V_{I}}{I k}&- \frac{(1-c \eta_{1})\alpha T V_{I}}{I}\\
 (1-c \eta_{1})\alpha V_{I} & 0& -d_{2}-d_{3}&0&0\\
 0&(1-c \eta_{1})\alpha I&0&a_{55}&(1-c \eta_{1})\alpha T\\
 0&0&-(\frac{\eta_{r}+\eta_{1}}{2})\frac{\beta I}{V_{I}}&(\frac{2-\eta_{r}-\eta_{1}}{2})\beta & a_{66}\\
  \end{array}
\right].$\\

We now define  Lozinski$\breve{i}$ measure of matrix B as follows:
 \begin{equation}\label{eqn-6-hep}
\mu(B)\leq \max \{g_{1},g_{2}\}
\end{equation}
where $g_{1} = \mu(B_{11})+ || B_{12} ||$ and $g_{2} =|| B_{21} ||+ \mu(B_{22})$ $(||. ||$ denotes the vector norm).\\\\

It can be shown that (see Appendix)
 \begin{equation}\label{eqn-29-hep}
\mu(B)\leq \frac{\dot{I}}{I}-\overline{b}
\end{equation}
 where
\begin{eqnarray}\label{eqn-30-hep}
\nonumber \overline{b}&=&\frac{\beta I}{V_{I}}-\beta-(1-c \eta_{1})\alpha V_{1}- 2 (1-c \eta_{1})\alpha T-\left(\frac{\eta_{r}+\eta_{1}}{2}\right)\frac{2 \beta I}{V_{I}}\\
&-& \max \left\{-d_{2}+\left(\frac{2-\eta_{r}-\eta_{1}}{2}\right)\beta, -d_{3}+(1-c \eta_{1})\alpha T \right\}
\end{eqnarray}
for sufficiently large t. Then along each solution $(T(t),I(t), V_{I}(t),V_{NI}(t))$ such that $(T(0),I(0), V_{I}(0),V_{NI}(0)) \in E $ we have for $ t > \overline{t}$
\begin{equation}\label{eqn-31-hep}
 \frac{1}{t}\int_{0}^{t} \mu(B)ds \leq \frac{1}{t}\int_{0}^{\overline{t}} \mu(B)ds+\frac{1}{t}\ln\left(\frac{I(t)}{I(\overline{t})}\right)-\left(\frac{t-\overline{t}}{t}\right)\eta.
\end{equation}
Thus boundedness of $I(t)$ and definition of $\overline{q}_{2}$ finally gives $\overline{q}_{2}< 0.$ The endemic equilibrium point $E_{*}(T^{*},I^{*},V_{I}^{*},V_{NI}^{*})$ is globally stable.
\end{proof}

\section{Numerical Results}

 System (1-4) has been simulated for viral kinetics during successful antiviral treatment with the parameter values taken from table 1. Fig.2 shows the kinetics of the viral decline in patients responding to interferon, which is biphasic in nature. During the first phase, almost all patients treated with interferon show rapid dose dependent viral decline for about 24 to 48 hours \cite{Neum98,Lam97,Zeu98}. A second phase starts after about 48 hours where the viral decline is slow. This dynamics is well captured in fig.2, which shows that the viral load declines to a very low level and eventually eradicated, depending on the efficacy of interferon.\\

 When the effectiveness of interferon is quite large, ribavirin does not have significant impact on the second phase decline \cite{Feld05}. This is evident from fig.3. Comparing fig.2 with fig.3, it is observed that both of them portray almost identical graphs though the efficacy of ribavirin in fig.2 is zero ($\eta_r = 0$) and in fig.3 is 0.3 ($\eta_r = 0.3$), confirming the fact that ribavirin fails to alter the viral load when the efficacy of interferon is effectively large.\\

 It has been reported that there are few patients who do not respond to initial interferon therapy and are termed as null-responders. Re-treatment of null-responders to a standard interferon regimen shows a considerable first phase decline, followed by a minor further decline during the second phase and then a rebound of HCV is observed \cite{Bekk97,Bekk98}. This dynamics is well portrayed in fig.4, which corresponds to the rebound of the viral load during standard interferon regime to the re-treatment of non-responders.\\

 However, ribavirin enhances the second phase slope if the effectiveness of interferon is significantly less that 1 (see fig.5). This interesting behavior is due to the fact that interferon initially clears the infectious virions effectively, which results in the first phase decline of the viral load. Ribavirin does not have any effect on the first phase decline as it does not alter the viral load. Also, when the effectiveness of interferon is large, ribavirin has very little role to play as virion production is low but when interferon effectiveness is small, ribavirin renders progeny virions non-infectious and enhances the second phase decline. This behavior has been confirmed by experimental observations \cite{Herrmann03,Layden03,Dahari04}. Fig.5 is able to capture this dynamics with lower interferon efficacy and varying the effectiveness of ribavirin. From the figure it is concluded that the patient attains sustained virological responses (SVR) during the combined drug therapy, that is, the concentration of the viral particles in the blood falls below the detection limit during therapy and remains undetected for 24 weeks even after the therapy is stopped. Patients who exhibit SVR are generally cured of the HCV infection.  \\

 Several experiments have been conducted to evaluate ribavirin monotherapy in the treatment of chronic HCV \cite{Di92,Di95,Bod97}. In all these experiments, no virologic end-of-treatment-responses (ETR) were observed; either there is a minor decline or HCV RNA level remains constant after 3-6 months of ribavirin monotherapy compared with pretreatment values \cite{Di92,Di95}. This dynamics is reflected in fig6. \\

 Fig.7 shows the viral kinetics after the therapy cessation. A virus resurgence post therapy cessation has been simulated for the system (1-4), with different drug efficacies $(\eta_1,\eta_r) = ((0.4,0.6); (0.4,0.8); (0.4,0.99))$ from time 0 to 14 days and then set $\eta_1$ and $\eta_r$ to zero for the rest of the simulation. After 14 days, the virus resurges to pretreatment levels within (7-7.5) days of post therapy cessation. Thus, the model given by system (1-4), predicts resurgences to pretreatment levels after cessation of therapy.\\

 The triphasic decay of viral load for certain parameter values is illustrated in fig.8. As mentioned earlier, viral production and hence viral load decreases due to interferon action during the first phase; as a result of which, the production of infected cells by new infections falls and the total number of cells decline. Homeostatic mechanisms act to restore the total number of cells by cell proliferation. However, since the proliferation term is only with uninfected cells, homeostatic mechanisms result predominantly in the proliferation of the uninfected cells. But the rate at which the uninfected cells gets infected is much less than the rate at which the immune mediated effective killing of virions continues due to interferon effectiveness and this results in the second phase of the triphasic decline. As the effectiveness of interferon gradually decreases, the role of ribavirin becomes predominant and results in further declining of the viral load, which explains the third phase of the triphasic decline of the viral load (see fig.8). \\

 As there is an increase in the influx rate of new hepatocytes(s), the triphasic viral decay disappears and yields a biphasic decay (see fig.9). It is also to be noted that for different drug efficacies close to 1, both the biphasic viral decay (fig.2) and the triphasic viral decay (fig.8) is close to the death rate ($d_2$) of infected hepatocytes (I).

\section{Conclusion}

Millions of people all over the world are infected with HCV with high levels of mortality \cite{NIH02}. The models illustrating the HCV dynamics provide key insights into the HCV pathogenesis in vivo and action mechanism of interferon and ribavirin \cite{Perelson05}. To study the dynamics of hepatitis C virus, a mathematical model comprising of four coupled ordinary differential equations has been proposed to understand the role of ribavirin in interferon therapy. The positivity and boundedness of the system has been established. The basic reproduction number ($R_0$) and the controlled reproduction number ($\widehat{R}$) of the model have been calculated. It is observed mathematically that the local stability of the uninfected steady state depends on basic reproduction number ($R_0$) and infected steady state on controlled reproduction number ($\widehat{R}$). The permanence and global stability analysis of the model are established. A critical drug efficacy component $\eta_c$ has been defined and characterized.\\

Numerical simulation of system (1) shows some interesting results, which mimic the dynamics of hepatitis C virus in patients with successful therapy with interferon and ribavirin. There is a rapid decline in viral load followed by a second slower decline (biphasic decline) until the virus becomes undetectable \cite{Neum98,Colom03,Paw04}. A triphasic decay \cite{Dahari07b} of viral load has also been observed in the model with proper choice of interferon and ribavirin effectiveness. The critical drug efficacy parameter leads to a point (a transcritical bifurcation point) in the model that determines the successful antiviral therapy leading to HCV eradication or leading only to a partial response. Thus, the HCV dynamics that the model predicts, provide some insights on the possible mechanism for the behavior of the viral load observed in the clinic \cite{Dahari07b}. We sincerely hope that the application of this model will give a better understanding in the success of anti-viral therapy in HCV infection.

\section{Appendix}
\textbf{Calculation of Lozinski$\breve{i}$ measure of matrix B}

Lozinski$\breve{i}$ measure of matrix B has been defined as follows:
 \begin{equation}\label{eqn-6-hep}
\mu(B)\leq \max \{g_{1},g_{2}\}
\end{equation}
where $g_{1} = \mu(B_{11})+ || B_{12} ||$ and $g_{2} =|| B_{21} ||+ \mu(B_{22})$ $(||. ||$ denotes the vector norm).

It is easy to compute $\mu(B_{11}) = -\frac{r T}{k}-\frac{s}{ T}-d_{2}$, $||B_{12}||=(1-c \eta_{1})\alpha T $ and $||B_{21}||=\beta $, then
\begin{eqnarray}
\label{eqn-7-hep}
  g_{1}&=& -\frac{r T}{k}-\frac{s}{ T}-d_{2}+(1-c \eta_{1})\alpha T\\
\label{eqn-8-hep}
  g_{2}&=& \beta + \mu(B_{22}=C)
\end{eqnarray}
The matrix $[B_{22}]_{5\times5}$ is partitioned as
$B_{22}= C = \left[
  \begin{array}{cc}
  C_{11}& C_{12}\\
  C_{21}& C_{22}
   \end{array}
\right]$\\
where
$ C_{11}=[a_{22}]=[-\frac{r T}{k}-\frac{s}{ T}-d_{3}],C_{12}=\left[
  \begin{array}{cccc}
    0& -\frac{r T}{k}& 0 & 0\\
       \end{array}
\right]$,
$C_{21} = [0 (1-c \eta_{1})\alpha V_{I} 0 0]^T$ and
$C_{22} = \left[
  \begin{array}{cccc}
 a_{22}&0&-\frac{r T V_{I}}{I k}&- \frac{(1-c \eta_{1})\alpha T V_{I}}{I}\\
    0& -d_{2}-d_{3}&0&0\\
 (1-c \eta_{1})\alpha I&0&a_{55}&(1-c \eta_{1})\alpha T\\
 0&-(\frac{\eta_{r}+\eta_{1}}{2})\frac{\beta I}{V_{I}}&(\frac{2-\eta_{r}-\eta_{1}}{2})\beta & a_{66}\\
  \end{array}
\right]$.\\\\
Accordingly,
\begin{equation}\label{eqn-9-hep}
\mu(C)\leq \max \{g_{3},g_{4}\}
\end{equation}
where 
\begin{eqnarray}
\label{eqn-10-hep}
g_{3} = \mu(C_{11})+ || C_{12}|| = -\frac{s}{ T}-d_{3}\\
\label{eqn-11-hep}
g_{4} =(1-c \eta_{1})\alpha V_{I})+ \mu(C_{22})
\end{eqnarray}
To compute $\mu(C_{22})$, the matrix $[C_{22}]_{5\times5}$ is again partitioned as
$C_{22}= D = \left[
  \begin{array}{cc}
  D_{11}& D_{12}\\
  D_{21}& D_{22}
   \end{array}
\right]$,
where
$ D_{11}=[a_{22}] = [-\frac{r T}{k}-\frac{s}{ T}-d_{3}],
D_{12}=\left[
  \begin{array}{cccc}
    0& -\frac{r T V_{I}}{I k}&- \frac{(1-c \eta_{1})\alpha T V_{I}}{I}\\
       \end{array}
\right]$,\\\\
$D_{21} = [0~~ (1-c \eta_{1})\alpha T~~ 0]^T$ and
$D_{22} = \left[
  \begin{array}{ccc}
     -d_{2}-d_{3}&0&0\\
     0&a_{55}&(1-c \eta_{1})\alpha T\\
 -(\frac{\eta_{r}+\eta_{1}}{2})\frac{\beta I}{V_{I}}&(\frac{2-\eta_{r}-\eta_{1}}{2})\beta & a_{66}\\
  \end{array}
\right]$\\\\
Denoting
\begin{equation}
\label{eqn-12-hep}
\mu(D)\leq \max \{g_{5},g_{6}\}
\end{equation}
 where $g_{5} = \mu(D_{11})+ || D_{12}||$ and $g_{6} =|| D_{21}||+ \mu(D_{22})$.\\\\
We have, $\mu(D_{11}) = -\frac{r T}{k}-\frac{s}{ T}-d_{3}$,\\\\
$||D_{12}||= \frac{T V_{I}}{I}$ max $\left[\frac{r}{k}, (1-c \eta_{1})\alpha\right]$\\\\
$     = \frac{T V_{I}}{I} [(1-c \eta_{1})\alpha]$ as $(1-c \eta_{1})\alpha > \frac{r}{k}$, $||D_{21}||= (1-c \eta_{1})\alpha T.$ Let

$D_{22}= E = \left[
  \begin{array}{cc}
  E_{11}& E_{12}\\
  E_{21}& E_{22}
   \end{array}
\right]$, then
\begin{eqnarray}
\label{eqn-13-hep}
  g_{5}&=&  -\frac{r T}{k}-\frac{s}{ T}-d_{3}+ \frac{T V_{I}}{I} [(1-c \eta_{1})\alpha]\\
\label{eqn-14-hep}
  g_{6}&=& (1-c \eta_{1})\alpha T+ \mu(D_{22}=E)
\end{eqnarray}
and\\
$ E_{11}=[-d_{2}-d_{3}],
E_{12}=\left[
  \begin{array}{cc}
    0&0\\
       \end{array}
\right]$,
$E_{21} = \left[
  \begin{array}{c}
    0\\
  -(\frac{\eta_{r}+\eta_{1}}{2})\frac{\beta I}{V_{I}}\\
     \end{array}
\right]$\\\\
$E_{22} = \left[
  \begin{array}{ccc}
          a_{55}&(1-c \eta_{1})\alpha T\\
 (\frac{2-\eta_{r}-\eta_{1}}{2})\beta & a_{66}\\
  \end{array}
\right]$\\\\
Again, we define Lozinski$\breve{i}$ measure of E as follows:
\begin{equation}\label{eqn-15-hep}
\mu(E)\leq \max \{g_{7},g_{8}\}
\end{equation}
where $g_{7} = \mu(E_{11})+ || E_{12} ||$ and $g_{8} =|| E_{21}||+ \mu(E_{22}).$\\\\
We have, $\mu(E_{11}) = -d_{2}-d_{3}$, $||E_{12}||=0$, $||E_{21}||=\left(\frac{\eta_{r}+\eta_{1}}{2}\right)\frac{\beta I}{V_{I}}$ and
\begin{eqnarray}
 \nonumber \mu(E_{22})&=& \max \left[a_{55}+\left(\frac{2-\eta_{r}-\eta_{1}}{2}\right)\beta, (1-c \eta_{1})\alpha T+a_{66}\right] \\
\nonumber    &=& \left\{\frac{\dot{I}}{I}-\frac{\dot{V}_{I}}{V_{I}}-d_{3}\right\}+\max \left\{-d_{2}+\left(\frac{2-\eta_{r}-\eta_{1}}{2}\right)\beta, -d_{3}+(1-c \eta_{1})\alpha T \right\}
   \end{eqnarray}
Therefore

\begin{eqnarray}\label{eqn-16-hep}
  g_{7}&=&  -d_{2}-d_{3}\\
\label{eqn-17-hep}
\nonumber  g_{8}&=& \left\{\frac{\dot{I}}{I}-\frac{\dot{V}_{I}}{V_{I}}-d_{3}\right\}+\left(\frac{\eta_{r}+\eta_{1}}{2}\right)\frac{\beta I}{V_{I}}\\
  & + & \max \left\{-d_{2}+\left(\frac{2-\eta_{r}-\eta_{1}}{2}\right)\beta, -d_{3}+(1-c \eta_{1})\alpha T \right\}
\end{eqnarray}
So, from equation (\ref{eqn-16-hep}), (\ref{eqn-17-hep}) and inequality (\ref{eqn-15-hep}), we get
\begin{eqnarray}\label{eqn-18-hep}
\nonumber \mu(E)& \leq & \left\{\frac{\dot{I}}{I}-\frac{\dot{V}_{I}}{V_{I}}-d_{3}\right\}+\left(\frac{\eta_{r}+\eta_{1}}{2}\right)\frac{\beta I}{V_{I}}\\
  & + & \max \left\{-d_{2}+\left(\frac{2-\eta_{r}-\eta_{1}}{2}\right)\beta, -d_{3}+(1-c \eta_{1})\alpha T \right\}
\end{eqnarray}
From (\ref{eqn-14-hep}) and (\ref{eqn-18-hep}), we obtain,
\begin{eqnarray}\label{eqn-19-hep}
\nonumber g_{6} &\leq & \left\{\frac{\dot{I}}{I}-\frac{\dot{V}_{I}}{V_{I}}-d_{3}\right\}+2 (1-c \eta_{1})\alpha T+\left(\frac{\eta_{r}+\eta_{1}}{2}\right)\frac{\beta I}{V_{I}}\\
&+& \max \left\{-d_{2}+\left(\frac{2-\eta_{r}-\eta_{1}}{2}\right)\beta, -d_{3}+(1-c \eta_{1})\alpha T \right\}
\end{eqnarray}
Also,
\begin{eqnarray}
\label{eqn-20-hep}
  \frac{\dot{V_{I}}}{V_{I}} &=& \left(\frac{2-\eta_{r}-\eta_{1}}{2}\right)\frac{\beta I}{V_{I}}-d_{3} \\
\label{eqn-21-hep}
 \frac{\dot{I}}{I} &=& \frac{(1-c \eta_{1})\alpha T V_{I}}{I}-d_{2}
\end{eqnarray}
Then from (\ref{eqn-19-hep}) and (\ref{eqn-20-hep}), we get,
\begin{eqnarray}\label{eqn-22-hep}
\nonumber g_{6} &\leq& \frac{\dot{I}}{I}-\frac{\beta I}{V_{I}}+ 2 (1-c \eta_{1})\alpha T+\left(\frac{\eta_{r}+\eta_{1}}{2}\right)\frac{2 \beta I}{V_{I}} \\
&+& \max \left\{-d_{2}+\left(\frac{2-\eta_{r}-\eta_{1}}{2}\right)\beta, -d_{3}+(1-c \eta_{1})\alpha T  \right\}
\end{eqnarray}
From (\ref{eqn-13-hep}) and (\ref{eqn-21-hep}), we get,
\begin{equation}\label{eqn-23-hep}
g_{5}=   \frac{\dot{I}}{I} +d_{2}-\frac{r T}{k}-\frac{s}{ T}-d_{3}
\end{equation}
From (\ref{eqn-12-hep}),(\ref{eqn-22-hep}) and (\ref{eqn-23-hep}), we get,
\begin{eqnarray}\label{eqn-24-hep}
\nonumber \mu(D)& \leq & \frac{\dot{I}}{I}-\frac{\beta I}{V_{I}}+ 2 (1-c \eta_{1})\alpha T+\left(\frac{\eta_{r}+\eta_{1}}{2}\right)\frac{2 \beta I}{V_{I}} \\
&+& \max \left\{-d_{2}+\left(\frac{2-\eta_{r}-\eta_{1}}{2}\right)\beta, -d_{3}+(1-c \eta_{1})\alpha T  \right\}
\end{eqnarray}
From (\ref{eqn-11-hep}) and (\ref{eqn-24-hep}), we get,
\begin{eqnarray}\label{eqn-25-hep}
\nonumber g_{4} &\leq & \frac{\dot{I}}{I}-\frac{\beta I}{V_{I}}+(1-c \eta_{1})\alpha V_{1}+ 2 (1-c \eta_{1})\alpha T+\left(\frac{\eta_{r}+\eta_{1}}{2}\right)\frac{2 \beta I}{V_{I}} \\ &+& \max \left\{-d_{2}+\left(\frac{2-\eta_{r}-\eta_{1}}{2}\right)\beta, -d_{3}+(1-c \eta_{1})\alpha T  \right\}
\end{eqnarray}
From (\ref{eqn-9-hep}), (\ref{eqn-10-hep}) and (\ref{eqn-25-hep}), we get,
\begin{eqnarray}\label{eqn-26-hep}
\nonumber\mu(C) &\leq & \frac{\dot{I}}{I}-\frac{\beta I}{V_{I}}+(1-c \eta_{1})\alpha V_{1}+ 2 (1-c \eta_{1})\alpha T+\left(\frac{\eta_{r}+\eta_{1}}{2}\right)\frac{2 \beta I}{V_{I}} \\ &+& \max \left\{-d_{2}+\left(\frac{2-\eta_{r}-\eta_{1}}{2}\right)\beta, -d_{3}+(1-c \eta_{1})\alpha T  \right\}
\end{eqnarray}
From (\ref{eqn-8-hep}) and (\ref{eqn-26-hep}), it follows
\begin{eqnarray}\label{eqn-27-hep}
\nonumber g_{2} & \leq & \frac{\dot{I}}{I}-\frac{\beta I}{V_{I}}+\beta+(1-c \eta_{1})\alpha V_{1}+ 2 (1-c \eta_{1})\alpha T+\left(\frac{\eta_{r}+\eta_{1}}{2}\right)\frac{2 \beta I}{V_{I}}\\ &+& \max \left\{-d_{2}+\left(\frac{2-\eta_{r}-\eta_{1}}{2}\right)\beta, -d_{3}+(1-c \eta_{1})\alpha T \right\}
\end{eqnarray}
From (\ref{eqn-7-hep}) and (\ref{eqn-21-hep}), we obtain
\begin{equation}\label{eqn-28-hep}
g_{1}=\frac{\dot{I}}{I}-\frac{r T}{k}-\frac{s}{ T}-\frac{(1-c \eta_{1})\alpha V_{I} T}{I}+(1-c \eta_{1})\alpha T
\end{equation}
Therefore, we get
 \begin{equation}\label{eqn-29-hep}
\mu(B)\leq \frac{\dot{I}}{I}-\overline{b}~~\textrm{for sufficiently large t.}
\end{equation}
where
\begin{eqnarray}\label{eqn-30-hep}
\nonumber \overline{b}&=&\frac{\beta I}{V_{I}}-\beta-(1-c \eta_{1})\alpha V_{1}- 2 (1-c \eta_{1})\alpha T-\left(\frac{\eta_{r}+\eta_{1}}{2}\right)\frac{2 \beta I}{V_{I}}\\
&-& \max \left\{-d_{2}+\left(\frac{2-\eta_{r}-\eta_{1}}{2}\right)\beta, -d_{3}+(1-c \eta_{1})\alpha T \right\}
\end{eqnarray}

\begin{acknowledgements}
This study was supported by the Initiation Grant A (Grant number IITR/SRIC/100518) from Indian Institute of Technology, Roorkee, India.
\end{acknowledgements}


%


\begin{table}
\caption{Parameter values used for numerical calculations}
\label{tab:1}       
\begin{tabular}{lll}
\hline\noalign{\smallskip}
Parameters & Values  \\
\noalign{\smallskip}\hline\noalign{\smallskip}
$s$ (constant rate uninfected hepatocytes production) & $1.0 ~cell/ml/day$ \cite{Dahari07a}. \\\hline
$r$ (proliferation rate) & $1.99 /day$ \cite{Dahari07a}; $(0.47-3.7) /day$ \cite{Daheri09}.\\\hline
  $k$ (carrying capacity) & $3.6\times 10^7 cells /ml$ \cite{Dahari07a}. \\\hline
  $d_1$ (natural death of hepatocytes) & 0.01/day \cite{Dahari07a}. \\\hline
  $\alpha$ (rate of infection of uninfected hepatocytes) & $2.25\times 10^{-7} ml/day/virions$ \cite{Dahari07a}; $(1.5-27.8)\times 10^{-7} ml/day/virions$ \cite{Daheri09}.\\\hline
  $d_2$ (natural death rate of infected hepatocytes) & $1.0 /day$ \cite{Dahari07a}; (0.43 - 3.1)/day \cite{Daheri09}. \\\hline
  $\beta$ (rate at which virions are replicated) & $2.9~virions/cell/day$ \cite{Dahari07a}; $(1.2-7.9)~virions/cell/day$ \cite{Daheri09}. \\\hline
  $d_3$ (natural death rate of infected and non infected virions) & $6.0 /day$ \cite{Dahari07a}. \\
\noalign{\smallskip}\hline
\end{tabular}
Note: The effectiveness of the drugs $\eta_1$ and $\eta_r$ lies between 0 and 1. They are chosen accordingly and are mentioned in the figures. The value of c also lies between 0 and 1.
\end{table}

\begin{figure}[!hbtp]
\centering
\includegraphics[width=5.5in]{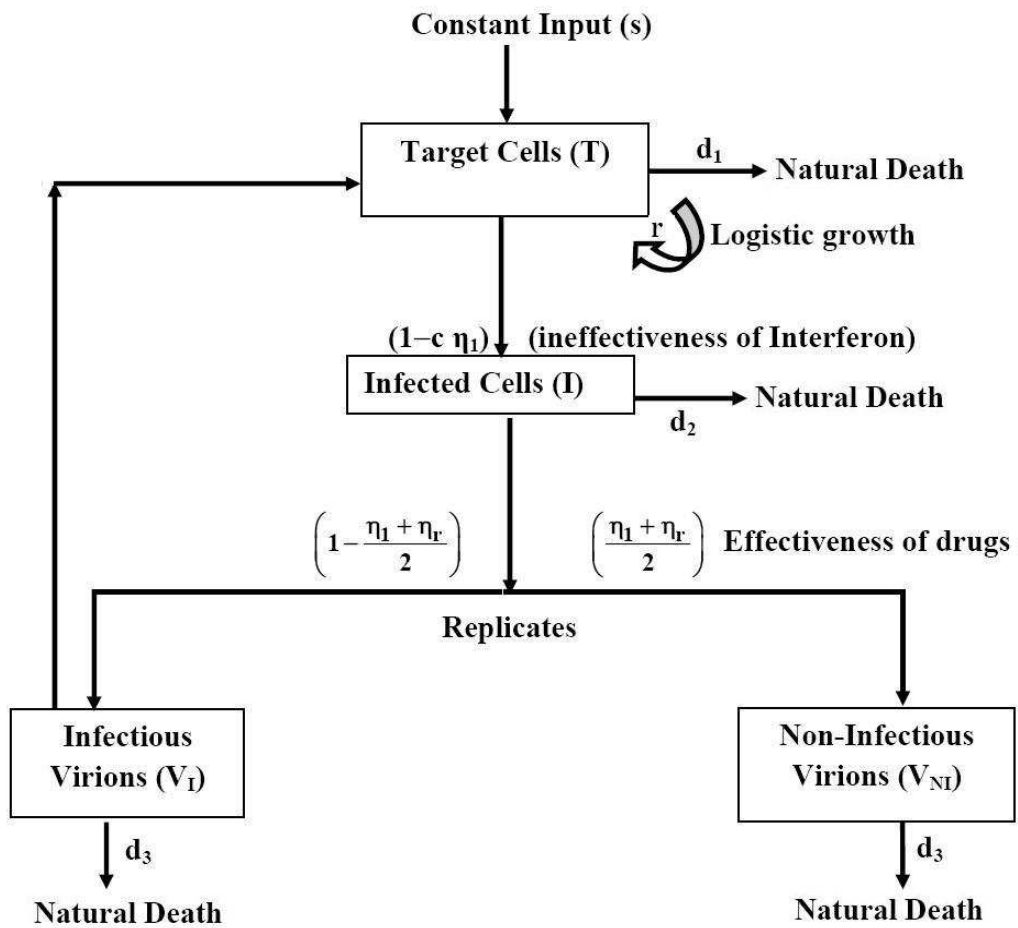}
\caption{\emph{The figure shows the schematic diagram, explaining the dynamics of Hepatitis C virus infection.}}
\end{figure}

\begin{figure}[!hbtp]
\centering
\includegraphics[width=4.0in]{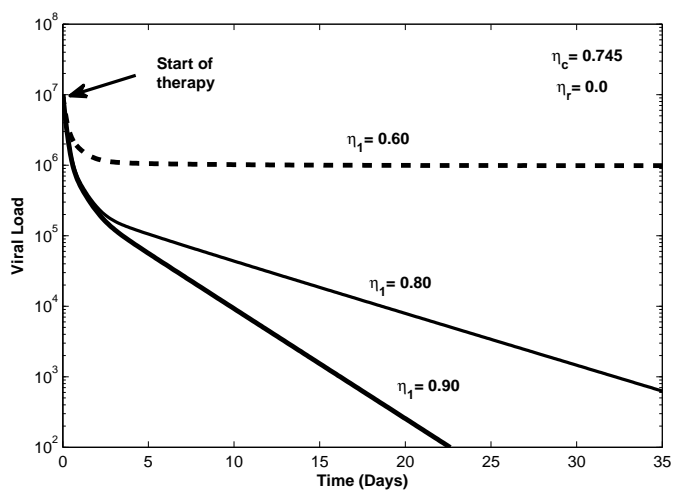}
\caption{\emph{The figure shows the kinetics of viral decline in patients responding to interferon only, which is characterized by biphasic graphs. The effect of ribavirin ($\eta_r$) is taken to be zero in this case. Three biphasic decline has been shown for various drug efficiencies, of which, two of them are greater than the critical drug efficiency $\eta_c$. All the parameter values are taken from Table 1.  }}
\end{figure}

\begin{figure}[!hbtp]
\centering
\includegraphics[width=4.0in]{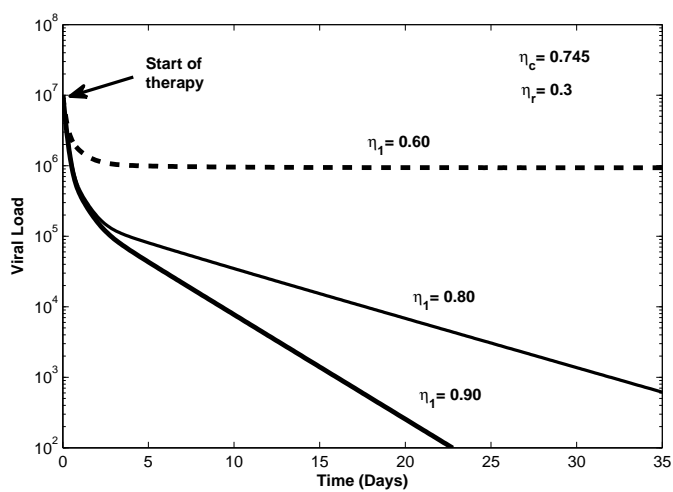}
\caption{\emph{The figure shows that ribavirin has no role to play in reducing the viral load, when the effectiveness of interferon is large. This is evident as viral load has not been altered when compared to figure 2, where effectiveness of ribavirin is zero. All the parameter values are taken from Table 1.}}
\end{figure}

\begin{figure}[!hbtp]
\centering
\includegraphics[width=4.0in]{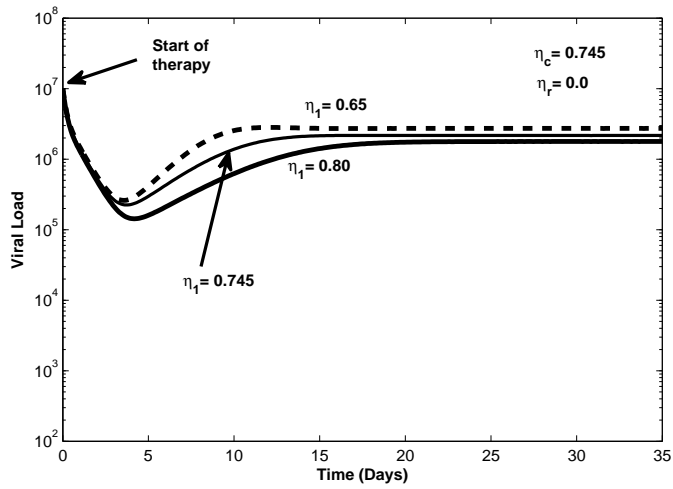}
\caption{\emph{The figure shows the re-treatment of non-responders to a standard interferon regime, which shows a first phase decline, followed by further minor decline during the second phase and then a rebound of viral load, which are sometimes observed in patients. All the parameter values are taken from Table 1.}}
\end{figure}

\begin{figure}[!hbtp]
\centering
\includegraphics[width=4.0in]{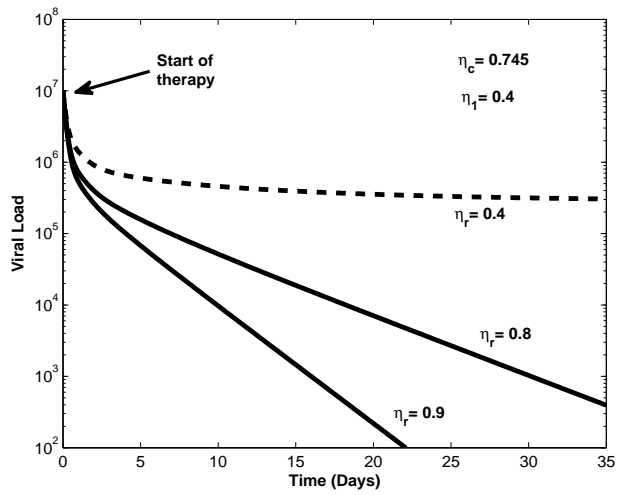}
\caption{\emph{The figure shows the biphasic decline of the viral load when the effectiveness of interferon is small. All the parameter values are taken from Table 1.}}
\end{figure}

\begin{figure}[!hbtp]
\centering
\includegraphics[width=4.0in]{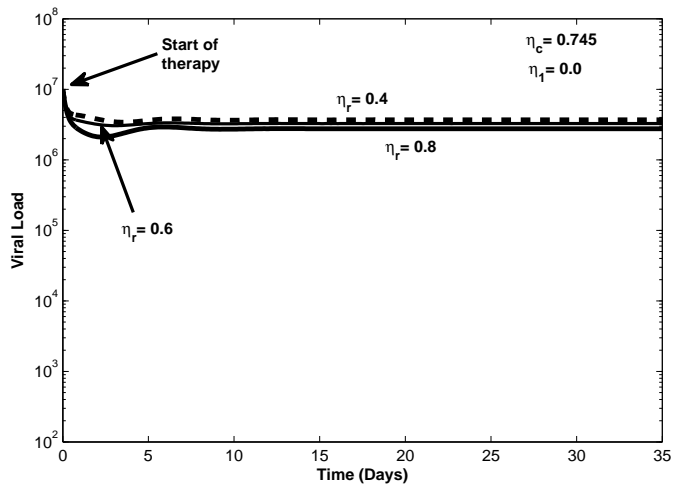}
\caption{\emph{The figure shows the dynamics of ribavirin monotherapy in the treatment of chronic HCV. There is very little decline in the viral load, even with high doses of ribavirin. All the parameter values are taken from Table 1.}}
\end{figure}

\begin{figure}[!hbtp]
\centering
\includegraphics[width=4.0in]{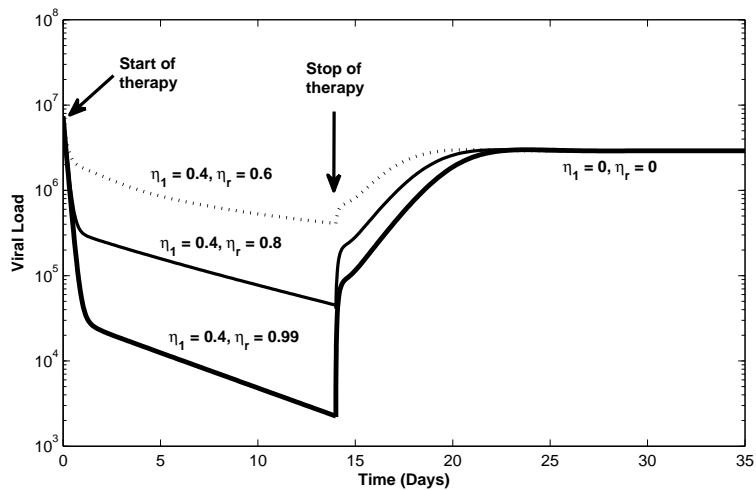}
\caption{\emph{The figure shows the viral kinetics after the therapy cessation. With different drug efficiencies, the decline in viral load is observed from 0 to 14 days. After 14 days, the drug efficiencies are set to zero and the virus resurges to pretreatment levels within (7-7.5) days of post therapy cessation. All the parameter values are taken from Table 1.}}
\end{figure}

\begin{figure}[!hbtp]
\centering
\includegraphics[width=4.0in]{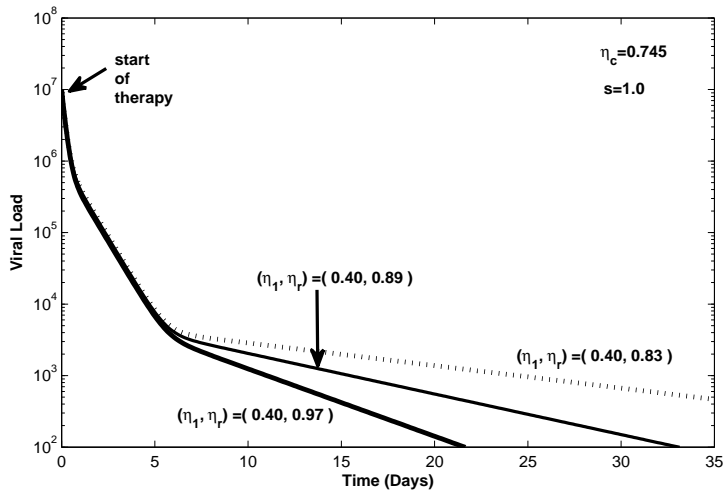}
\caption{\emph{The figure shows the dynamics of triphasic decline in the viral load. All the parameter values are taken from Table 1.}}
\end{figure}

\begin{figure}[!hbtp]
\centering
\includegraphics[width=4.0in]{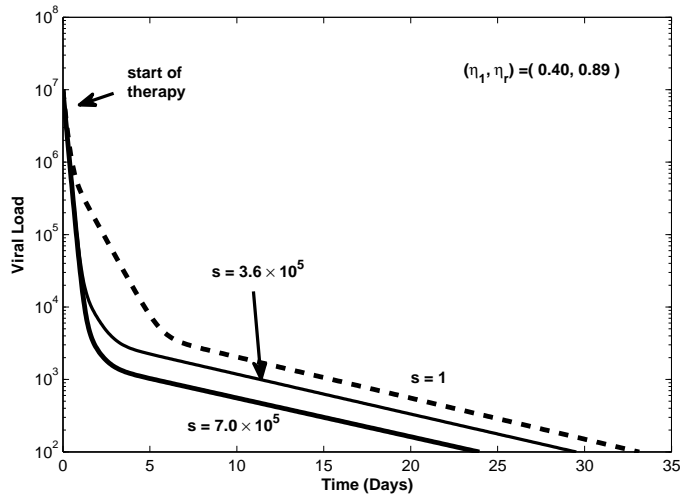}
\caption{\emph{The figure shows the effect of higher influx rates of new hepatocytes, where the triphasic decline of viral load changes to biphasic decline. All the parameter values are taken from Table 1.}}
\end{figure}

\end{document}